\documentclass[11pt]{article}

\usepackage[letterpaper, margin=1in]{geometry}

\usepackage{amsmath,amssymb,amsthm,mathtools,bm}

\usepackage{booktabs}
\usepackage{graphicx}

\usepackage{enumitem}

\usepackage{natbib}

\usepackage[dvipsnames]{xcolor}

\usepackage[colorlinks=true]{hyperref}

\newcommand{\E}{\mathbb{E}}
\newcommand{\Pp}{\mathbb{P}}
\newcommand{\R}{\mathbb{R}}
\newcommand{\G}{\mathcal{G}}
\newcommand{\F}{\mathcal{F}}
\newcommand{\D}{\mathcal{D}}
\newcommand{\A}{\mathcal{A}}
\newcommand{\given}{\mid}

\newcommand{\expit}{\operatorname{expit}}
\newcommand{\argmax}{\operatorname*{arg\,max}}
\newcommand{\ind}{\mathbf{1}}

\theoremstyle{plain}
\newtheorem{theorem}{Theorem}
\newtheorem{proposition}{Proposition}
\newtheorem{corollary}{Corollary}

\theoremstyle{definition}
\newtheorem{definition}{Definition}
\newtheorem{example}{Example}

\theoremstyle{remark}
\newtheorem{remark}{Remark}

\title{Practical Boundary Degeneracy and Reverse-Martingale Limits\\
       in Sequential Binary Models}

\author{Yuan-chin Ivan Chang\\[4pt]
  \small Institute of Statistical Science, Academia Sinica\\
  \small 128 Academia Road, Section 2, Nankang, Taipei 11529, Taiwan\\
  \small \href{mailto:ivan@stat.sinica.edu.tw}{\texttt{ycchang@as.edu.tw}}}

\date{\small\today}

\begin{document}

%% hypersetup here avoids \ProcessKeyvaOptions errors at package-load time
\hypersetup{
  citecolor  = MidnightBlue,
  linkcolor  = MidnightBlue,
  urlcolor   = MidnightBlue,
  pdftitle   = {Practical Boundary Degeneracy and Reverse-Martingale Limits in Sequential Binary Models},
  pdfauthor  = {Yuan-chin Ivan Chang},
  pdfsubject = {Statistics},
  pdfkeywords= {Bernoulli trial, boundary probability, complete separation, confidence sequence, logistic regression, reverse martingale, sequential inference}
}

\maketitle

\begin{abstract}
A run of all failures, a run of all successes, or complete separation in a
logistic regression each tempts the analyst to declare a probability of exactly
zero or one.  The central message of this paper is that all three phenomena
share a common structure: finite sequential data justify practical boundary
statements of the form $p\leq\varepsilon$ or $p\geq1-\varepsilon$, but not
exact boundary probabilities.  The paper's contribution is to unify these three
settings under a single reverse-martingale framework and to derive a stopping
rule, $\tau_{\mathrm{RM}}$, that requires three conditions
simultaneously---boundary closeness $B_n\leq\varepsilon$, uncertainty width
$W_n\leq w$, and trajectory stability $r_n\leq\eta$---rather than boundary
closeness alone.  The reverse-martingale view recasts boundary degeneracy as a
property of the limiting conditional law $M_\infty=\E(Y\given\G_\infty)$
rather than a finite-sample event, complementing classical one-sided binomial
tests and Wald's sequential probability ratio test without replacing them.
Numerical studies across Bernoulli rare-event trials, low- and
high-dimensional logistic regression, controlled risk trajectories, and a real
health-economics data set demonstrate that boundary closeness alone is an
unreliable stopping signal, and that the stability condition separates
transient apparent certainty from genuine limiting degeneracy.
\end{abstract}



\smallskip
\noindent\textit{Submitted for publication.  Comments welcome.}

\section{Introduction}

Sequential binary data frequently create the appearance of degeneracy.  A safety monitor observes a long run of non-events; a logistic regression separates outcomes perfectly; a sequential risk score approaches zero before the underlying state has stabilized.  Each tempts the analyst to conclude that a probability is exactly zero or one.  The central statistical point of this paper is:
\[
\begin{aligned}
&\text{finite sequential data may justify practical boundary statements,}\\
&\text{but not exact boundary probabilities.}
\end{aligned}
\]

These three phenomena---Bernoulli all-failure runs, logistic separation, and near-degenerate risk trajectories---are treated as unrelated technical inconveniences in their respective literatures.  The contribution of this paper is to unify them.  All three are manifestations of the same confusion between a finite near-boundary estimate and the limiting conditional law to which the underlying process converges.  Once that identification is made, a common remedy follows: a stopping rule that requires boundary closeness, uncertainty control, \emph{and} trajectory stability simultaneously, rather than boundary closeness alone.

The mathematical vehicle is the classical reverse-martingale convergence theorem \citep{doob1953,durrett2019}.  For a decreasing filtration $\G_1\supseteq\G_2\supseteq\cdots$ and $Y\in\{0,1\}$, the process $M_n=\E(Y\given\G_n)$ converges almost surely to $M_\infty=\E(Y\given\G_\infty)$.  Boundary degeneracy is a property of $M_\infty$, not of any finite $M_n$.  The paper applies this observation to give a unified interpretation of the three settings and to derive the three-condition stopping rule $\tau_{\mathrm{RM}}$.

The classical tools---one-sided binomial tests and confidence bounds \citep{clopper1934}, confidence sequences \citep{howard2021}, and Wald's SPRT \citep{wald1945,wald1947}---each address one part of the picture and are used here as building blocks.  They are not replaced.  In particular, SPRT is the correct tool when the problem is genuinely simple-versus-simple Bernoulli monitoring; it does not cover composite one-sided statements, logistic separation, or trajectory stability, which is where the present framework adds content.  A brief technical comparison is given in Appendix~\ref{app:sprt_benchmark}.

The paper is organized as follows.  Section~\ref{sec:bernoulli} gives the Bernoulli boundary rules.  Section~\ref{sec:logistic} covers logistic separation.  Section~\ref{sec:rm} develops the reverse-martingale formulation and its illustrated example.  Section~\ref{sec:stopping} states the three-condition stopping rule.  Section~\ref{sec:applications} discusses sequential prediction and adaptive treatment.  Section~\ref{sec:numerical} presents numerical evidence, and Section~\ref{sec:discussion} concludes.

\section{Forward sequential boundary statements for Bernoulli trials}
\label{sec:bernoulli}

Let
\[
    Y_1,Y_2,\ldots \overset{\mathrm{iid}}{\sim} \operatorname{Bernoulli}(p),
    \qquad 0\leq p\leq 1,
\]
and define $S_n=\sum_{i=1}^n Y_i$.  The ordinary maximum-likelihood estimator is
\[
    \hat p_n=\frac{S_n}{n}.
\]
If $S_n=0$, then $\hat p_n=0$; if $S_n=n$, then $\hat p_n=1$.  These finite-sample equalities are estimates, not proofs that the data-generating probability is exactly zero or one.

\subsection{All-failure and all-success stopping}

The following elementary calculation gives a useful benchmark for practical degeneracy.

\begin{proposition}[All-failure practical-zero rule]
\label{prop:allfailure}
Fix $\varepsilon\in(0,1)$ and $\alpha\in(0,1)$.  Consider the rule that declares $p<\varepsilon$ after observing $n$ consecutive failures.  If
\[
    (1-\varepsilon)^n\leq \alpha,
\]
then
\[
    \sup_{p\geq \varepsilon}
    \Pp_p(Y_1=\cdots=Y_n=0)\leq \alpha.
\]
Equivalently, it is sufficient that
\[
    n\geq \frac{\log \alpha}{\log(1-\varepsilon)}.
\]
\end{proposition}

\begin{proof}
For every $p\geq \varepsilon$,
\[
    \Pp_p(Y_1=\cdots=Y_n=0)=(1-p)^n\leq (1-\varepsilon)^n\leq \alpha.
\]
This proves the claimed type-I error bound.
\end{proof}

\begin{corollary}[All-success practical-one rule]
\label{cor:allsuccess}
If $n$ consecutive successes are observed and $(1-\varepsilon)^n\leq \alpha$, then the rule declaring $p>1-\varepsilon$ has type-I error at most $\alpha$ under $p\leq 1-\varepsilon$.
\end{corollary}

Proposition~\ref{prop:allfailure} and Corollary~\ref{cor:allsuccess} are intentionally simple.  They show that a finite run can justify a one-sided practical statement with specified error control.  They do not justify $p=0$ or $p=1$.

\begin{remark}[Classical comparators]
\label{rem:classical}
The all-failure rule above is the $s=0$ special case of the standard one-sided binomial test of $H_0:p\geq\varepsilon$ against $H_1:p<\varepsilon$, with exact $p$-value $(1-\varepsilon)^n$ and dual upper Clopper--Pearson bound $U_n(0;\alpha)=1-\alpha^{1/n}$ \citep{clopper1934,lehmann2005}.  For sequential problems where the stopping time is data-dependent, fixed-time intervals are replaced by confidence sequences valid under optional stopping \citep{howard2021}.  When the question is instead a simple-versus-simple comparison between two pre-specified probabilities, Wald's SPRT is the canonical tool \citep{wald1945,wald1947}.  None of these classical methods addresses the composite, model-based, or trajectory-stability aspects of the boundary problem studied here; the present paper adds content in those directions rather than replacing the classical tools.  A brief SPRT comparison is in Appendix~\ref{app:sprt_benchmark}.
\end{remark}

\subsection{Smoothing and time-uniform uncertainty}

A common stabilized estimate arises from a beta prior,
\[
    p\sim \operatorname{Beta}(a,b),
    \qquad a,b>0.
\]
After observing $S_n$ successes in $n$ trials,
\[
    p\given \D_n\sim \operatorname{Beta}(a+S_n,b+n-S_n),
    \qquad
    \D_n=\sigma(Y_1,\ldots,Y_n),
\]
and the posterior mean is
\[
    \tilde p_n=\E(p\given \D_n)=\frac{S_n+a}{n+a+b}.
\]
This estimate is strictly inside $(0,1)$ for every finite $n$ when $a,b>0$.

For sequential inference, the more important object is not the smoothed point estimate but a stopping-valid uncertainty statement.  A confidence sequence is a sequence of random intervals $(C_n)$ satisfying
\[
    \Pp_p\{p\in C_n\text{ for all }n\geq 1\}\geq 1-\alpha,
    \qquad p\in[0,1].
\]
Such intervals are valid under arbitrary data-dependent stopping times; see, for example, \citet{wald1945} and \citet{howard2021}.  A practical-zero stopping time can be written as
\[
    \tau_0=\inf\{n:C_n\subseteq [0,\varepsilon]\},
\]
and a practical-one stopping time as
\[
    \tau_1=\inf\{n:C_n\subseteq [1-\varepsilon,1]\}.
\]
On the event that the confidence sequence covers $p$ at all times, $\tau_0$ cannot falsely declare practical zero when $p>\varepsilon$, and $\tau_1$ cannot falsely declare practical one when $p<1-\varepsilon$.

\begin{proposition}[Sequential validity of confidence-sequence stopping]
\label{prop:csstop}
Let $(C_n)$ be a $(1-\alpha)$ confidence sequence for $p$.  Define $\tau_0$ and $\tau_1$ as above.  Then
\[
    \Pp_p(\tau_0<\infty\text{ and }p>\varepsilon)\leq \alpha,
    \qquad
    \Pp_p(\tau_1<\infty\text{ and }p<1-\varepsilon)\leq \alpha.
\]
\end{proposition}

\begin{proof}
If $p\in C_n$ for all $n$, then $C_{\tau_0}\subseteq[0,\varepsilon]$ implies $p\leq\varepsilon$, and $C_{\tau_1}\subseteq[1-\varepsilon,1]$ implies $p\geq 1-\varepsilon$.  Therefore either type of false boundary declaration can occur only on the event that the confidence sequence fails to cover $p$ at some time, which has probability at most $\alpha$.
\end{proof}

\section{Logistic regression and separation}
\label{sec:logistic}

Let $(Y_i,x_i)$ be sequentially observed, where $Y_i\in\{0,1\}$ and $x_i\in\R^d$.  The logistic model is
\[
    \Pp(Y_i=1\given x_i)=\pi_i(\beta)=\expit(x_i^\top\beta)
    =\frac{\exp(x_i^\top\beta)}{1+\exp(x_i^\top\beta)}.
\]
The log-likelihood after $n$ observations is
\[
    \ell_n(\beta)=\sum_{i=1}^n\left\{y_i x_i^\top\beta-
    \log\bigl(1+\exp(x_i^\top\beta)\bigr)\right\}.
\]

\begin{proposition}[No exact boundary probabilities for finite logistic coefficients]
\label{prop:logisticinterior}
For every finite $\beta\in\R^d$ and every covariate vector $x\in\R^d$,
$0<\expit(x^\top\beta)<1$.
Consequently, fitted probabilities equal to zero or one cannot arise from a finite logistic-regression coefficient vector.
\end{proposition}

This is immediate from $\exp(x^\top\beta)\in(0,\infty)$ for all finite arguments.

Complete or quasi-complete separation is the logistic-regression counterpart of apparent Bernoulli boundary behavior.  If there exists $a\in\R^d$ such that
\[
    x_i^\top a>0 \quad \text{for all }y_i=1,
    \qquad
    x_i^\top a<0 \quad \text{for all }y_i=0,
\]
then moving $\beta$ along $ta$ with $t\to\infty$ drives fitted probabilities toward one for successes and zero for failures.  The likelihood may approach its supremum without attaining it at any finite $\beta$.  This is the classical existence problem for logistic maximum likelihood under separation \citep{albert1984}.

A stable sequential analysis should therefore regularize the coefficient estimate and base stopping on predictive probabilities and uncertainty, not on divergent coefficient estimates.  Standard choices include Firth's bias-reduced likelihood \citep{firth1993,heinze2002}, ridge-penalized logistic regression, and Bayesian logistic regression with a proper prior such as a normal or weakly informative Student prior \citep{gelman2008}.  For example, ridge logistic regression estimates
\[
    \hat\beta_{n,\lambda}=\argmax_{\beta\in\R^d}
    \left\{\ell_n(\beta)-\frac{\lambda}{2}\|\beta\|_2^2\right\},
    \qquad \lambda>0,
\]
which remains finite under separation.

For a covariate region $A\subseteq\R^d$, a practical stopping rule should be phrased in terms of the predictive surface
\[
    \pi(x;\beta)=\expit(x^\top\beta).
\]
A Bayesian practical-one rule, for example, is
\[
\begin{aligned}
    \tau_1(A)=\inf\Bigl\{n:
    &\Pp\left(\inf_{x\in A}\pi(x;\beta)>1-\varepsilon\given\D_n\right)\\
    &\geq 1-\alpha\Bigr\},
\end{aligned}
\]
and the corresponding practical-zero rule is
\[
\begin{aligned}
    \tau_0(A)=\inf\Bigl\{n:
    &\Pp\left(\sup_{x\in A}\pi(x;\beta)<\varepsilon\given\D_n\right)\\
    &\geq 1-\alpha\Bigr\}.
\end{aligned}
\]
Frequentist analogues replace posterior probabilities by simultaneous one-sided confidence bounds over $A$.

\section{Reverse-martingale formulation of boundary degeneracy}
\label{sec:rm}

The forward sequential view treats boundary behavior as a finite-sample decision under uncertainty.  The reverse-martingale view treats boundary behavior as a limiting property of conditional probabilities under a decreasing information structure.

\begin{theorem}[Conditional-probability reverse martingale]
\label{thm:rm}
Let $Y\in L^1$ and let
$
    \G_1\supseteq \G_2\supseteq \cdots
$
be a decreasing filtration.  Define
$
    M_n=\E(Y\given\G_n),
    \qquad n\geq 1.
$
Then $(M_n,\G_n)$ is a reverse martingale; that is,
$
    \E(M_n\given\G_{n+1})=M_{n+1},~ n\geq 1.
$
Moreover, if $\G_\infty=\cap_{n\geq1}\G_n$, then
$
    M_n\longrightarrow M_\infty=\E(Y\given\G_\infty)
$
almost surely and in $L^1$.
\end{theorem}

\begin{proof}
Because $\G_{n+1}\subseteq\G_n$, the tower property gives
\[
    \E(M_n\given\G_{n+1})=
    \E\{\E(Y\given\G_n)\given\G_{n+1}\}
    =\E(Y\given\G_{n+1})=M_{n+1}.
\]
The convergence assertion is the reverse-martingale convergence theorem; see \citet{doob1953} or \citet{durrett2019}.
\end{proof}

For binary $Y$, $0\leq M_n\leq1$ and $M_n$ has the interpretation of a conditional success probability.  Theorem~\ref{thm:rm} therefore turns boundary behavior into a statement about the limiting conditional law.

\begin{example}[Reverse martingale for a Bernoulli outcome under coarsening]
\label{ex:rm_bernoulli}
Let $Y\in\{0,1\}$ with $\Pp(Y=1)=p=0.1$, and suppose $n=5$ binary observations
$Y_1,\ldots,Y_5\overset{\mathrm{iid}}{\sim}\operatorname{Bernoulli}(0.1)$ are drawn.
Consider the \emph{decreasing} filtration
\[
    \G_n = \sigma(Y_n,Y_{n+1},\ldots,Y_5),
    \qquad n=5,4,3,2,1,
\]
so that $\G_5\subseteq\G_4\subseteq\cdots\subseteq\G_1$ and
$\G_\infty=\bigcap_{n=1}^{5}\G_n=\G_5=\sigma(Y_5)$.

\medskip
\noindent\textbf{Trajectory.}
Suppose the realized sequence is $(Y_5,Y_4,Y_3,Y_2,Y_1)=(0,0,0,0,1)$, i.e.\ four
failures followed (in reverse reading) by one success.  At each stage:

\begin{center}
\renewcommand{\arraystretch}{1.35}
\footnotesize
\begin{tabular}{ccll}
\toprule
Step $n$ & Information in $\G_n$ & $M_n=\E(Y\mid\G_n)$ & Comment\\
\midrule
5 & $\{Y_5=0\}$ & $\Pp(Y=1\mid Y_5=0)=p=0.1$ &
  $Y_5$ independent of $Y$; no update\\
4 & $\{Y_4=0,Y_5=0\}$ & $0.1$ &
  Still all failures; posterior unchanged\\
3 & $\{Y_3=0,Y_4=0,Y_5=0\}$ & $0.1$ &
  Apparent zero run grows\\
2 & $\{Y_2=0,\ldots,Y_5=0\}$ & $0.1$ &
  Run of 4 failures; $M_2$ stays at $p$\\
1 & $\{Y_1=1,Y_2=0,\ldots,Y_5=0\}$ & $0.1$ &
  One success revealed; still $p$\\
$\infty$ & $\G_\infty=\sigma(Y_5)$ & $M_\infty=p=0.1$ & Limit is the true probability\\
\bottomrule
\end{tabular}
\end{center}

\noindent
Because $Y_1,\ldots,Y_5$ are i.i.d.\ and independent of $Y$ (they are side-observations, not
$Y$ itself), each $\G_n$ carries no information about $Y$ beyond what is contained in the
marginal law.  Consequently $M_n=\E(Y\mid\G_n)=p$ for every $n$, and the reverse
martingale is constant: the trajectory does not approach a boundary.

\medskip
\noindent\textbf{Boundary scenario.}
Now replace the setup so that $Y$ is the \emph{first future trial}, and let
$\G_n=\sigma(Y_{n+1},\ldots,Y_{n+k})$ for a sliding look-ahead window.  If the window
repeatedly shows only failures ($S_k=0$), the posterior under a
$\operatorname{Beta}(a,b)$ prior updates to
\[
    M_n = \E(Y\mid\G_n)
    = \frac{a}{a+b+k}
    \xrightarrow[k\to\infty]{} 0.
\]
With $a=b=0.5$ (Jeffreys prior) and $k$ failures in the window:

\begin{center}
\renewcommand{\arraystretch}{1.35}
\begin{tabular}{rll}
\toprule
Failures $k$ & $M_n=(0.5)/(1+k)$ & Boundary distance $B_n=M_n$\\
\midrule
0  & $0.500$ & $0.500$\\
1  & $0.333$ & $0.333$\\
4  & $0.100$ & $0.100$\\
9  & $0.050$ & $0.050$\\
19 & $0.025$ & $0.025$\\
$\infty$ & $0$ & $0$ (boundary limit)\\
\bottomrule
\end{tabular}
\end{center}

\noindent
The sequence $(M_n)$ is a bounded reverse martingale converging to $M_\infty=0$, illustrating
\emph{exact boundary degeneracy} in the reverse-martingale limit (Definition~\ref{def:rmdeg}).
For the practical-degeneracy criterion with $\varepsilon=0.05$, the trajectory enters the
boundary zone at $k=9$ failures, i.e.\ once $M_n\leq\varepsilon$.

\medskip
\noindent\textbf{Reverse-martingale property check.}
Direct verification: for $n<n'$ (so $\G_{n'}\subseteq\G_n$),
\[
    \E(M_n\mid\G_{n'})
    =\E\{\E(Y\mid\G_n)\mid\G_{n'}\}
    =\E(Y\mid\G_{n'})
    =M_{n'},
\]
by the tower property, confirming that $(M_n,\G_n)$ satisfies the defining identity of
Theorem~\ref{thm:rm}.  The trajectory in the boundary scenario above is therefore a valid
reverse martingale; the limiting conditional law collapses toward zero, whereas no finite
$M_n$ is identically zero.
\end{example}

\begin{definition}[Reverse-martingale boundary degeneracy]
\label{def:rmdeg}
Let $Y\in\{0,1\}$ and $M_n=\E(Y\given\G_n)$.  The process is exactly boundary-degenerate in the reverse-martingale limit on an event $E$ if
\[
    M_\infty\in\{0,1\}
    \quad\text{on }E.
\]
For $\varepsilon\in(0,1/2)$, it is practically boundary-degenerate on $E$ if
\[
    M_\infty\leq\varepsilon
    \quad\text{or}\quad
    M_\infty\geq1-\varepsilon
    \quad\text{on }E.
\]
\end{definition}

This definition deliberately concerns $M_\infty$, not a finite $M_n$.  A finite conditional probability may be close to zero or one because of transient data imbalance, early separation, or model instability.  The reverse-martingale question is whether the limiting conditional law is moving toward a boundary.

\begin{remark}[Relation to increasing filtrations]
Sequential observation is usually modelled with an increasing filtration $\F_1\subseteq\F_2\subseteq\cdots$.  The decreasing filtration here arises naturally when the object of interest is a sequence of coarsened or backward-read conditional descriptions, so that the relevant information sets shrink and the limiting object is the intersection $\G_\infty$.
\end{remark}

\section{Boundary stopping with uncertainty and stability}
\label{sec:stopping}

A practical decision rule should distinguish three situations:
\begin{enumerate}[label=(\roman*)]
    \item interior uncertainty, where the conditional probability is far from both zero and one;
    \item unstable apparent certainty, where the conditional probability is close to a boundary but the trajectory is not stable;
    \item stable practical degeneracy, where the conditional probability is close to a boundary and the trajectory has stabilized.
\end{enumerate}

Let
$ B_n=\min\{M_n,1-M_n\} $
be the boundary distance.  Boundary closeness alone corresponds to $B_n\leq\varepsilon$.  To avoid stopping on transient extremes, one also needs an uncertainty measure and, in computed models, a stability diagnostic.

Suppose that $(I_n)$ is a time-uniform interval sequence for the limiting boundary target $M_\infty$:
\[
    \Pp\{M_\infty\in I_n\text{ for all }n\geq1\}\geq1-\alpha.
\]
The construction of such intervals is model-specific: it may come from a Bayesian posterior calculation, a bootstrap or subsampling scheme with appropriate calibration, or a problem-specific martingale bound.  Once such intervals are available, the stopping logic is immediate.

\begin{proposition}[Stopping validity for a limiting boundary target]
\label{prop:mstop}
Let $(I_n)$ satisfy
\[
    \Pp\{M_\infty\in I_n\text{ for all }n\geq1\}\geq1-\alpha.
\]
Define
\[
    \tau_0=\inf\{n:I_n\subseteq[0,\varepsilon]\},
    \qquad
    \tau_1=\inf\{n:I_n\subseteq[1-\varepsilon,1]\}.
\]
Then
\[
    \Pp(\tau_0<\infty\text{ and }M_\infty>\varepsilon)\leq\alpha,
    \qquad
    \Pp(\tau_1<\infty\text{ and }M_\infty<1-\varepsilon)\leq\alpha.
\]
\end{proposition}

\begin{proof}
The proof is identical to Proposition~\ref{prop:csstop}, replacing $p$ by $M_\infty$.
\end{proof}

In computational implementations, let $r_n\geq0$ denote a stability defect.  Examples include
\[
    r_n=|M_{n+1}-\widehat{\E}(M_n\given\G_{n+1})|
\]
when conditional expectations are estimated directly, or
\[
    r_n=\|H_n-g_\phi(H_{n+1})\|
\]
when a latent state $H_n$ is regularized by a backward projection $g_\phi$.  A practical boundary-stopping rule can then be written as
\[
    \tau_{\mathrm{RM}}
    =\inf\left\{n\geq n_{\min}:
        B_n\leq\varepsilon,
        \quad r_n\leq\eta,
        \quad W_n\leq w
    \right\},
\]
where $W_n$ is an uncertainty width and $\eta,w$ are pre-specified thresholds.  The rule operationalizes the core message: a near-boundary estimate becomes actionable only when it is statistically supported \emph{and} dynamically stable.  Its validity in any specific application depends on the calibration of $W_n$ and $r_n$, which are problem-specific.

\begin{remark}[Finite estimates versus limiting laws]
The reverse-martingale approach does not prove exact probabilities zero or one from finite data.  Its value is to separate a finite near-boundary estimate from a limiting boundary law.  This distinction is important in sequential applications because temporary separation and stable boundary behavior can look similar if only the current point estimate is inspected.
\end{remark}

\section{Numerical studies}
\label{sec:numerical}

The numerical studies are designed to illustrate the distinction between apparent finite-sample boundary behavior and practically justified boundary statements.  The purpose is not to compare classifiers by accuracy alone, but to examine when ordinary estimators produce boundary-looking outputs, when sequential rules can support statements such as $p\leq\varepsilon$, and when additional stability conditions are needed to distinguish transient apparent certainty from stable practical degeneracy.  All scripts used to generate the tables and figures are supplied with the manuscript source.  For the two logistic-regression experiments and the real-data rare-event illustration, we use 1,000 Monte Carlo repetitions per setting; the Bernoulli and reverse-martingale-style studies use the larger replication counts specified in their scripts.  Full data-generating mechanisms, tuning constants, instability criteria, and reproducibility details are collected in Appendix~\ref{app:numerical_details}.

The studies have a deliberately diagnostic role.  Study~1 asks how often an ordinary Bernoulli estimator equals zero even when the true success probability is positive.  Studies~2 and~3 ask when logistic regression converts rare events into separation-like numerical behavior, first in a low-dimensional model and then in a higher-dimensional sparse model.  Study~4 asks whether boundary closeness alone is a reliable stopping signal for a sequential conditional-risk trajectory.  Study~5 uses a real rare-event data set to check that the phenomenon is not only a simulation artifact.  Together, the simulations and the real-data illustration explain why the proposed revisit is not merely philosophical: the boundary phenomenon appears repeatedly in finite-sample computation, and the proposed practical-degeneracy language gives a more reliable way to interpret it.

\subsection{Bernoulli rare-event experiments}

We first consider iid Bernoulli trials with $p\in\{0.10,0.01,0.005\}$.  This setting isolates the boundary phenomenon from any covariate or model-fitting issue.  For each value of $p$, we generated Bernoulli samples at several maximum sample sizes and evaluated the ordinary estimator $\hat p_n=S_n/n$, the Jeffreys-smoothed estimator
\[
    \tilde p_n^J=\frac{S_n+1/2}{n+1},
\]
and the all-failure practical-zero stopping rule of Proposition~\ref{prop:allfailure}.  The reported probabilities for $\hat p_n=0$ and the all-failure rule are exact binomial probabilities, with Monte Carlo checks generated by the accompanying script.

% Auto-generated by 01_bernoulli_rare_event.py
\begin{table}[t]
\centering
\caption{Bernoulli rare-event simulations at $n_{\max}=1000$. The practical-zero rule stops after $n_\varepsilon=\lceil\log(0.05)/\log(1-\varepsilon)\rceil$ consecutive failures. Exact probabilities are reported for the apparent zero MLE and for practical-zero stopping.}
\label{tab:bernoulli}
\begin{tabular}{cccccc}
\toprule
$p$ & $\varepsilon$ & $n_\varepsilon$ & $P(\hat p_{1000}=0)$ & $P(\tau_0\leq1000)$ & $E(\tilde p^J_{1000})$ \\ 
\midrule
0.100 & 0.010 & 299 & 1.748e-46 & 0.0000 & 0.1004 \\ 
0.100 & 0.005 & 598 & 1.748e-46 & 0.0000 & 0.1004 \\ 
0.010 & 0.010 & 299 & 4.317e-05 & 0.0495 & 0.0105 \\ 
0.010 & 0.005 & 598 & 4.317e-05 & 0.0025 & 0.0105 \\ 
0.005 & 0.010 & 299 & 0.006654 & 0.2234 & 0.0055 \\ 
0.005 & 0.005 & 598 & 0.006654 & 0.0499 & 0.0055 \\ 
\bottomrule
\end{tabular}
\end{table}

Table~\ref{tab:bernoulli} shows the key distinction.  When $p=0.005$ and $n=1000$, the ordinary MLE is exactly zero with probability about $0.0067$, even though the true probability is not zero.  When the tolerance is $\varepsilon=0.01$, the all-failure stopping rule has probability about $0.223$ under $p=0.005$, which is appropriate because $p$ is in fact below the practical-zero tolerance.  By contrast, when $p=0.01$ and $\varepsilon=0.005$, the same rule has probability only about $0.0025$.  Thus the rule can support a practical statement such as $p<\varepsilon$ with error control, but the finite equality $\hat p_n=0$ should not be interpreted as proof that $p=0$.

\begin{figure}[t]
\centering
\includegraphics[width=0.72\textwidth]{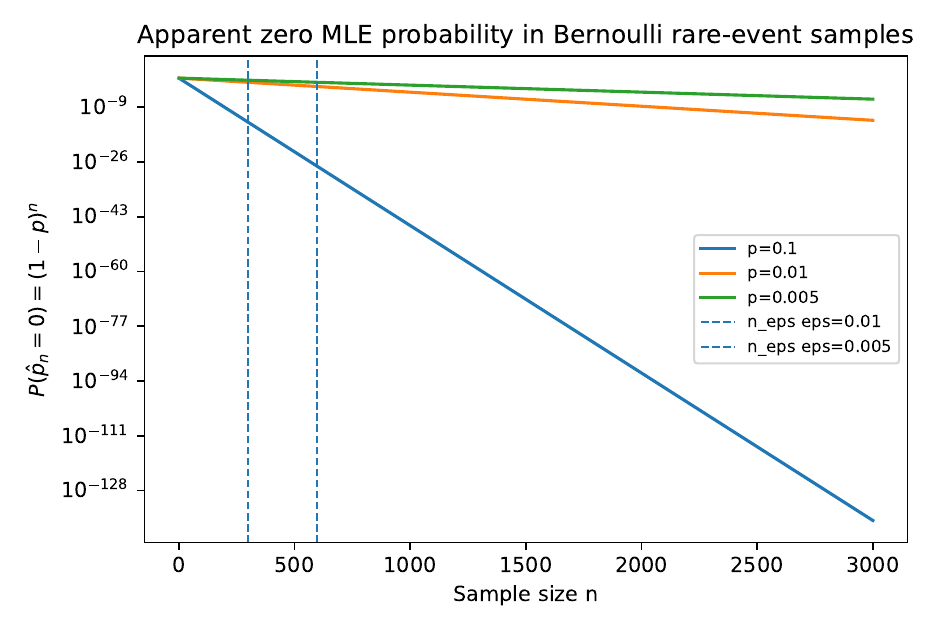}
\caption{Exact probability of an apparent zero Bernoulli MLE, $P(\hat p_n=0)=(1-p)^n$, as a function of sample size.  Vertical dashed lines mark the all-failure thresholds for $\alpha=0.05$ and $\varepsilon\in\{0.01,0.005\}$.}
\label{fig:bernoulli_zero}
\end{figure}

Figure~\ref{fig:bernoulli_zero} reinforces the same message graphically.  The apparent boundary estimate decays quickly when $p=0.10$, but can persist for rare events.  This is precisely the setting in which practical boundary statements are useful, provided they are phrased with a pre-specified tolerance and error level.

\subsection{Low-dimensional rare-event logistic experiments}

The second experiment introduces covariates while keeping the dimension small.  We generated
\[
    X_i\sim N(0,I_3),\qquad
    Y_i\mid X_i\sim \operatorname{Bernoulli}\{\expit(\beta_0+X_i^\top\beta)\},
\]
with $\beta=(1,-0.7,0.5)^\top$.  The intercept $\beta_0$ was calibrated by Monte Carlo integration so that $E\{\expit(\beta_0+X^\top\beta)\}=\rho$, where $\rho\in\{0.10,0.01,0.005\}$.  We compared ordinary logistic maximum likelihood with ridge-penalized logistic regression.  Ridge logistic regression is used here as a computationally simple representative of finite regularized estimators; Firth and Bayesian logistic regression play the same conceptual role in preventing divergent estimates under separation.

% Auto-generated by 02_lowdim_logistic.py; R=1000
\begin{table}[t]
\centering
\caption{Low-dimensional rare-event logistic simulations with three covariates and 1,000 Monte Carlo repetitions per setting. The table reports the mean number of events, one-class sample rate, ordinary-MLE instability rate, and median test log loss for ordinary MLE and ridge logistic regression.}
\label{tab:lowdim}
\begin{tabular}{ccccccc}
\toprule
Target $\rho$ & $n$ & Events & One-class & MLE unstable & Method & Median log loss \\ 
\midrule
0.100 & 100 & 9.8 & 0.000 & 0.019 & mle & 0.2699 \\ 
 &  &  &  &  & ridge & 0.2656 \\ 
0.100 & 500 & 50.0 & 0.000 & 0.000 & mle & 0.2673 \\ 
 &  &  &  &  & ridge & 0.2674 \\ 
0.100 & 2000 & 200.0 & 0.000 & 0.000 & mle & 0.2689 \\ 
 &  &  &  &  & ridge & 0.2689 \\ 
0.010 & 100 & 1.0 & 0.379 & 0.667 & mle & 0.0760 \\ 
 &  &  &  &  & ridge & 0.0559 \\ 
0.010 & 500 & 5.0 & 0.003 & 0.089 & mle & 0.0452 \\ 
 &  &  &  &  & ridge & 0.0448 \\ 
0.010 & 2000 & 19.7 & 0.000 & 0.000 & mle & 0.0477 \\ 
 &  &  &  &  & ridge & 0.0478 \\ 
0.005 & 100 & 0.5 & 0.596 & 0.801 & mle & 0.0487 \\ 
 &  &  &  &  & ridge & 0.0304 \\ 
0.005 & 500 & 2.5 & 0.088 & 0.326 & mle & 0.0295 \\ 
 &  &  &  &  & ridge & 0.0282 \\ 
0.005 & 2000 & 10.0 & 0.000 & 0.017 & mle & 0.0288 \\ 
 &  &  &  &  & ridge & 0.0285 \\ 
\bottomrule
\end{tabular}
\end{table}

Table~\ref{tab:lowdim} shows that rare events alone can create boundary-like logistic behavior.  At $\rho=0.01$ and $n=100$, the mean number of events is less than one, one-class samples occur frequently, and the ordinary logistic fit is unstable in a large fraction of repetitions.  At $\rho=0.005$, the same phenomenon is even more pronounced.  As $n$ increases, the average number of events grows and ordinary logistic regression becomes more stable.  Ridge regularization produces finite fitted probabilities in the same settings and gives more stable predictive log loss when events are sparse.

\begin{figure}[t]
\centering
\includegraphics[width=0.72\textwidth]{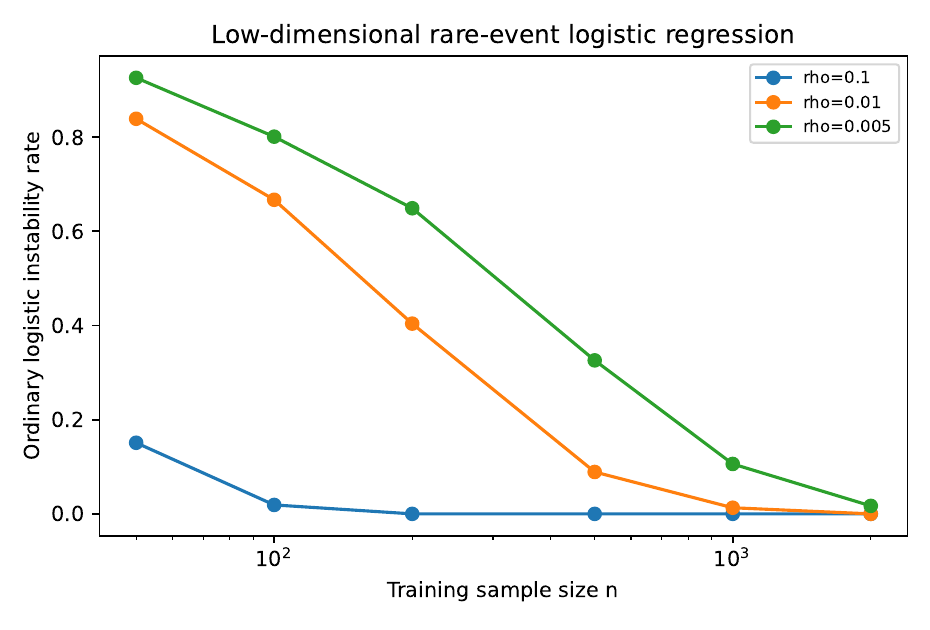}
\caption{Instability rate of ordinary logistic regression in the low-dimensional rare-event experiment.  Instability includes one-class samples, nonconvergence, excessively large coefficient norms, or fitted logits beyond a numerical boundary threshold.}
\label{fig:lowdim}
\end{figure}

Figure~\ref{fig:lowdim} shows that instability is not caused by high dimension alone.  Even with three covariates, logistic regression can exhibit separation-like behavior when the event probability is small and the sample size is modest.  This supports the interpretation in Section~\ref{sec:logistic}: apparent boundary probabilities in logistic regression should be read as separation, near-separation, or numerical divergence, not as finite-sample evidence for exact probabilities zero or one.

\subsection{High-dimensional rare-event logistic experiments}

The third experiment considers a more computationally challenging rare-event setting.  We generated $X_i\sim N(0,I_d)$ with $d\in\{20,50\}$ and used a sparse signal
\[
    \beta_j=\frac{2}{\sqrt{5}},\quad j=1,\ldots,5,
    \qquad
    \beta_j=0,
    \quad j>5.
\]
The intercept was calibrated to target prevalences $\rho\in\{0.01,0.005\}$.  We report results for $n=500$ and $n=1000$.  The summary includes the events-per-variable ratio, one-class sample rate, ordinary-MLE instability rate, and median test log loss for ordinary and ridge logistic regression.

% Auto-generated by 03_highdim_logistic.py; R=1000
\begin{table}[t]
\centering
\caption{High-dimensional rare-event logistic simulations with 1,000 Monte Carlo repetitions per setting. The table reports the mean number of events, events per variable (EPV), one-class sample rate, ordinary-MLE instability rate, and median test log loss for ordinary MLE and ridge fits.}
\label{tab:highdim}
\begin{tabular}{ccccccccc}
\toprule
$d$ & $\rho$ & $n$ & Events & EPV & One-class & MLE unstable & Method & Log loss \\ 
\midrule
20 & 0.010 & 500 & 5.0 & 0.250 & 0.007 & 0.962 & mle & 0.3886 \\ 
 &  &  &  &  &  &  & ridge & 0.0553 \\ 
20 & 0.010 & 1000 & 9.9 & 0.494 & 0.000 & 0.799 & mle & 0.0672 \\ 
 &  &  &  &  &  &  & ridge & 0.0555 \\ 
20 & 0.005 & 500 & 2.5 & 0.127 & 0.064 & 0.997 & mle & 0.1704 \\ 
 &  &  &  &  &  &  & ridge & 0.0260 \\ 
20 & 0.005 & 1000 & 5.0 & 0.252 & 0.006 & 0.957 & mle & 0.1754 \\ 
 &  &  &  &  &  &  & ridge & 0.0336 \\ 
50 & 0.010 & 500 & 5.0 & 0.100 & 0.006 & 1.000 & mle & 0.3263 \\ 
 &  &  &  &  &  &  & ridge & 0.0654 \\ 
50 & 0.010 & 1000 & 10.1 & 0.202 & 0.000 & 1.000 & mle & 0.5566 \\ 
 &  &  &  &  &  &  & ridge & 0.0627 \\ 
50 & 0.005 & 500 & 2.5 & 0.051 & 0.086 & 1.000 & mle & 0.1347 \\ 
 &  &  &  &  &  &  & ridge & 0.0369 \\ 
50 & 0.005 & 1000 & 5.1 & 0.102 & 0.003 & 1.000 & mle & 0.1995 \\ 
 &  &  &  &  &  &  & ridge & 0.0424 \\ 
\bottomrule
\end{tabular}
\end{table}

Table~\ref{tab:highdim} illustrates how rare outcomes and many variables amplify the boundary problem.  Even when the data-generating logistic model has finite coefficients, the observed sample may contain only a few events relative to the number of variables.  In such cases the ordinary logistic likelihood often behaves as if a separating direction is present, producing large fitted logits and poor out-of-sample log loss.  Ridge regularization stabilizes the fitted score and avoids many of these numerical extremes.

\begin{figure}[t]
\centering
\includegraphics[width=0.72\textwidth]{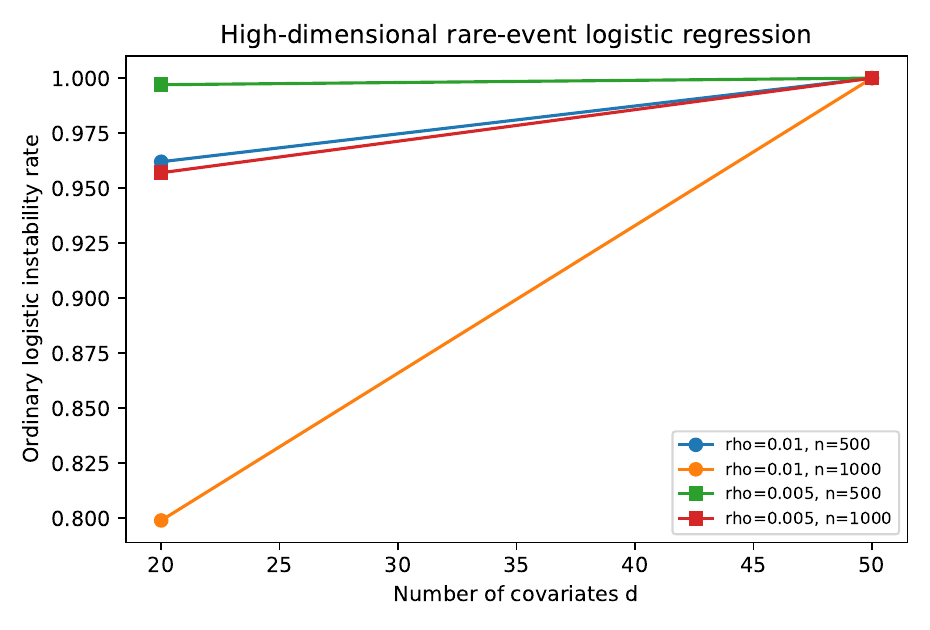}
\caption{Ordinary logistic instability in the high-dimensional rare-event experiment.  Instability becomes frequent when the number of events is small relative to the number of covariates.}
\label{fig:highdim}
\end{figure}

Figure~\ref{fig:highdim} summarizes the same pattern.  The instability rate is high for the smaller target prevalence and increases with the number of covariates.  The point is not that logistic regression is intrinsically inappropriate, but that ordinary unregularized logistic estimation should not be interpreted literally when sparse rare-event data create near-separation.

\subsection{Boundary closeness and reverse-martingale-style stability}

The final simulation isolates the role of the stability requirement in the reverse-martingale stopping rule.  We simulated three conditional-risk trajectories $M_t$: a stable boundary-convergent trajectory, a transient boundary trajectory that approaches zero temporarily and then returns to an interior value, and a stable interior trajectory.  For each path we computed the boundary distance
\[
    B_t=\min\{M_t,1-M_t\}
\]
and a simple backward-looking stability defect
\[
    r_t=\left|M_t-\frac{1}{h}\sum_{k=1}^h M_{t+k}\right|.
\]
We then compared a boundary-only rule, which stops when $B_t\leq\varepsilon$, with a boundary-plus-stability rule, which additionally requires $r_t\leq\eta$ and a decreasing uncertainty width $W_t\leq w$.

% Auto-generated by 04_rm_stability.py
\begin{table}[t]
\centering
\caption{Boundary closeness versus reverse-martingale-style stability. The boundary-only rule stops when $B_t\leq\varepsilon$. The stability rule additionally requires $r_t\leq\eta$ and $W_t\leq w$.}
\label{tab:rmstab}
\scriptsize
\begin{tabular}{p{0.40\textwidth}ccc}
\toprule
Scenario and rule & Stop probability & Mean stop time & Median stop time \\ 
\midrule
stable boundary, boundary only & 1.000 & 38.2 & 38.0 \\ 
stable boundary, boundary plus stability & 1.000 & 50.6 & 49.0 \\ 
transient boundary, boundary only & 0.999 & 48.4 & 48.0 \\ 
transient boundary, boundary plus stability & 0.107 & 47.8 & 48.0 \\ 
interior stable, boundary only & 0.000 & -- & -- \\ 
interior stable, boundary plus stability & 0.000 & -- & -- \\ 
\bottomrule
\end{tabular}
\end{table}

Table~\ref{tab:rmstab} shows that boundary closeness alone is insufficient.  In the transient-boundary scenario, the boundary-only rule stops almost always, although the risk trajectory later returns to an interior value.  Adding the stability condition sharply reduces such transient stopping while preserving stopping in the stable boundary-convergent scenario, at the cost of a modest delay.  This is the numerical counterpart of the reverse-martingale principle: a near-zero or near-one conditional risk estimate should be treated as actionable only when the trajectory is also stable.

\begin{figure}[t]
\centering
\includegraphics[width=0.72\textwidth]{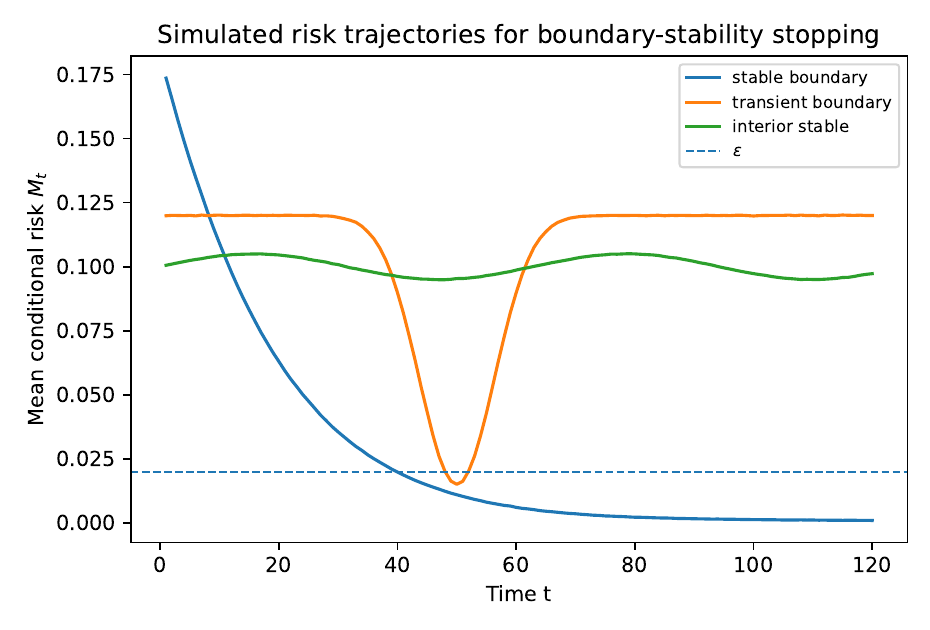}
\caption{Mean simulated conditional-risk trajectories in the reverse-martingale-style stability experiment.  The transient trajectory approaches the boundary temporarily but does not represent stable boundary behavior.}
\label{fig:rmstab}
\end{figure}

\subsection{Real-data rare-event illustration: RAND Health Insurance Experiment}
\label{sec:randhie}

As a real-data illustration, we use the RAND Health Insurance Experiment data included in the \texttt{statsmodels} data library.  The RAND HIE was a major randomized health-insurance study; the version used here contains 20,190 observations and 10 variables, and the \texttt{statsmodels} documentation identifies \texttt{hlthp} as an indicator for self-rated poor health \citep{statsmodels_randhie}.  In the full data, only 302 observations have \texttt{hlthp}=1, giving an empirical prevalence of approximately 1.50\%.  This makes the data useful for illustrating rare-event binary modelling in a real, non-simulated setting.

The outcome is
\[
    Y_i=\ind\{\text{self-rated health is poor}\}.
\]
We considered two covariate designs.  The base design uses six variables available in the data set: \texttt{lncoins}, \texttt{idp}, \texttt{lpi}, \texttt{fmde}, \texttt{physlm}, and \texttt{disea}.  The quadratic design augments these six standardized variables with all second-order terms and pairwise products, giving 27 predictors.  The latter design is not intended as a preferred substantive health model.  It is a controlled stress test showing what happens when an analyst fits a more flexible logistic score to a rare-event data set with few events per variable.

For each design and each training size $n\in\{200,500,1000\}$, we drew 1,000 random training subsamples without replacement.  We compared a nearly unpenalized logistic fit, used as a numerical proxy for ordinary maximum likelihood, with ridge logistic regression.  The nearly unpenalized fit uses a very weak $L_2$ penalty only to keep the optimization finite under separation-like samples; large fitted logits, extreme fitted probabilities, and convergence warnings are still interpreted as boundary instability.  For each replicate, we recorded the number of events, whether the training sample contained only one outcome class, the maximum absolute fitted logit, the fraction of fitted probabilities below $10^{-6}$ or above $1-10^{-6}$, and whether the fit met an instability criterion.

\begin{table}[t]
\centering
\caption{Real-data rare-event illustration using the RAND Health Insurance Experiment data.  The outcome is self-rated poor health, with prevalence 1.50\% in the full data.  Each row is based on 1,000 random training subsamples.  The quadratic design contains all first- and second-order terms of the six standardized base predictors.  The column labelled MLE uses a nearly unpenalized logistic fit with a very weak L2 penalty for numerical finiteness.}
\label{tab:randhie}
\scriptsize
\resizebox{\textwidth}{!}{%
\begin{tabular}{llrrrrrrrrr}
\toprule
Design & $n$ & $d$ & Mean events & One-class \% & MLE unstable \% & Ridge unstable \% & MLE max $|\eta|$ & Ridge max $|\eta|$ & MLE loss & Ridge loss \\
\midrule
base & 200 & 6 & 3.2 & 3.8 & 77.0 & 3.8 & 23.9 & 4.9 & 0.048 & 0.068 \\
base & 500 & 6 & 7.4 & 0.0 & 29.7 & 0.0 & 9.0 & 5.6 & 0.057 & 0.061 \\
base & 1000 & 6 & 14.9 & 0.0 & 3.5 & 0.0 & 7.3 & 5.9 & 0.061 & 0.063 \\
quadratic & 200 & 27 & 3.0 & 4.6 & 100.0 & 11.4 & 311.6 & 11.8 & 0.002 & 0.033 \\
quadratic & 500 & 27 & 7.4 & 0.0 & 100.0 & 6.0 & 430.5 & 12.8 & 0.021 & 0.044 \\
quadratic & 1000 & 27 & 15.1 & 0.0 & 99.6 & 1.9 & 169.6 & 11.9 & 0.045 & 0.054 \\
\bottomrule
\end{tabular}%
}
\end{table}

Table~\ref{tab:randhie} shows that the real data reproduce the same boundary phenomenon seen in the simulations.  With $n=200$, the expected number of poor-health events is only about three, and one-class training samples occur about 4\%--5\% of the time.  Even in the six-variable base design, the nearly unpenalized logistic fit is unstable in 77.0\% of repetitions at $n=200$ and 29.7\% at $n=500$.  The instability is much more severe in the quadratic design: the nearly unpenalized fit is unstable in essentially all repetitions for $n=200$ and $n=500$, with median maximum absolute logits far beyond ordinary probability interpretation.  Ridge regularization greatly reduces these numerical extremes, although the smallest samples remain difficult because the event count is intrinsically small.

\begin{figure}[t]
\centering
\includegraphics[width=0.72\textwidth]{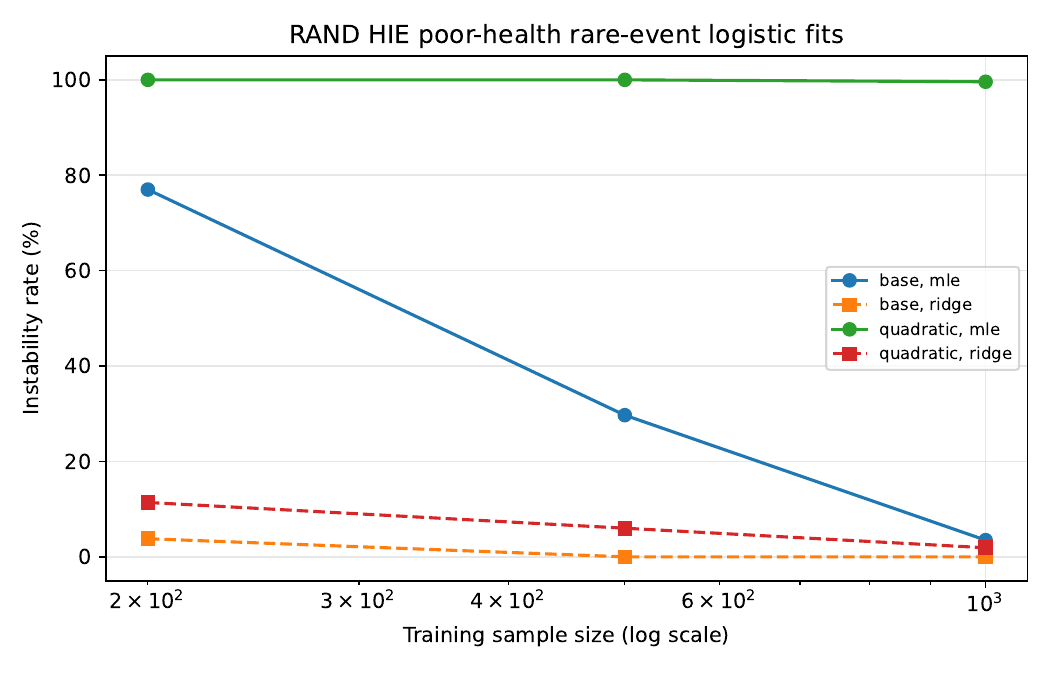}
\caption{Instability rate in the RAND HIE poor-health rare-event illustration.  The outcome prevalence is approximately 1.50\%.  The quadratic design amplifies separation-like behavior by increasing the number of fitted covariate terms relative to the number of events.}
\label{fig:randhie_instability}
\end{figure}

Figure~\ref{fig:randhie_instability} summarizes the main empirical lesson.  Rare-event data with a modest number of predictors can already produce boundary-looking fitted probabilities in small random subsamples; adding flexible terms makes the problem substantially worse.  This supports the practical message of the paper: an extreme fitted probability in a rare-event logistic analysis should be treated as a warning sign about event scarcity, separation, and model complexity, not as finite-sample proof of an exact zero or one probability.

\subsection{Synthesis of numerical findings}
\label{sec:numerical_synthesis}

The five numerical studies support the same conclusion from complementary directions.  In the Bernoulli experiment, an apparent zero estimate is a frequent finite-sample event when the true success probability is small.  For example, a positive but rare probability can generate an all-failure sample and hence \(\hat p_n=0\).  The correct interpretation is not that the data prove \(p=0\), but that a pre-specified rule may or may not support a practical statement such as \(p\leq\varepsilon\).  This is the simplest setting in which the distinction between exact and practical degeneracy is visible.

The logistic experiments show why the same distinction matters once covariates are introduced.  Even when the data-generating logistic model has finite coefficients and strictly interior probabilities, rare events can produce one-class samples, near-separating directions, large fitted logits, and numerically extreme fitted probabilities.  These effects occur in the low-dimensional experiment and become more severe in the high-dimensional rare-event experiment, where the number of events per variable is small.  The ridge fits are not introduced as a new optimal method; they serve as a computational demonstration that regularization changes the interpretation from divergent coefficients to finite predictive probabilities.

The real-data RAND HIE illustration confirms that the issue is not created only by artificial examples.  The poor-health outcome is rare, with about 1.50\% prevalence.  Random training subsamples therefore often contain only a few events, and the fitted logistic score can become extreme even though the data come from a real health-economics study rather than a constructed separation example.  The quadratic design further shows how routine modelling choices, such as adding flexible transformations and interactions, can turn rare-event scarcity into boundary-like numerical behavior.

The reverse-martingale-style experiment adds the sequential lesson.  A conditional-risk trajectory can be close to the boundary for a short period and then return to the interior.  A boundary-only rule treats this transient movement as decisive, whereas the boundary-plus-stability rule is more selective.  This supports the main operational message of the paper: a near-zero or near-one probability estimate should become actionable only when it is accompanied by uncertainty control and stability of the underlying conditional-risk process.

Thus, the numerical studies clarify the value of the proposed framework.  The revisit brings together three phenomena that are often discussed separately: zero estimates in Bernoulli rare-event data, separation in logistic regression, and premature stopping in sequential risk prediction.  The common remedy is not to deny boundary behavior, but to state it at the right level: finite samples justify controlled practical boundary statements, while reverse martingales describe boundary behavior as a limiting conditional-law property.

\section{Consequences for logistic and sequential prediction}
\label{sec:applications}

\subsection{Logistic prediction over covariate regions}

In separated logistic regression, the coefficient vector may diverge while the fitted probability surface approaches a meaningful limiting classifier.  For a covariate value $x$, write
\[
    M_n(x)=\Pp(Y=1\given X=x,\G_n).
\]
If the fitted logistic coefficient is finite, $M_n(x)=\expit\{x^\top\beta_n\}\in(0,1)$.  Under accumulating separation, however, $x^\top\beta_n$ may tend to $+\infty$ or $-\infty$, and $M_n(x)$ may tend to one or zero.  The stable object is therefore the limiting conditional probability surface, not necessarily the limiting coefficient vector.

For a clinically or operationally relevant region $A\subseteq\R^d$, define the region-wise boundary distance
\[
    B_n(A)=\sup_{x\in A}\min\{M_n(x),1-M_n(x)\}.
\]
The condition $B_n(A)\leq\varepsilon$ means that the entire region is close to one of the two boundaries.  A region-wise practical stopping rule has the form
\[
    \tau_{\mathrm{RM}}(A)
    =\inf\left\{n\geq n_{\min}:
        B_n(A)\leq\varepsilon,
        \quad \sup_{x\in A}r_n(x)\leq\eta,
        \quad W_n(A)\leq w
    \right\}.
\]
This guards against two errors: stopping because a coefficient is large but unstable, and stopping because a small early sample happens to be separated.

\subsection{Computational implementation through backward coherence}

The preceding framework is model-agnostic.  In low-dimensional logistic regression, regularized likelihood or Bayesian updating may be sufficient.  In high-dimensional longitudinal prediction, one may compute a latent state $H_t$ from patient or system histories $Z_{1:t}$ and impose a backward-coherence penalty
\[
    L_{\mathrm{RM}}=\frac{1}{T-1}\sum_{t=1}^{T-1}
    \|H_t-g_\phi(H_{t+1})\|^2.
\]
The corresponding defect
\[
    r_t=\|H_t-g_\phi(H_{t+1})\|
\]
acts as a stability diagnostic.  A binary risk output $M_t=\Pp(Y=1\given H_t)$ should be treated as actionable boundary evidence only when
\[
    \min\{M_t,1-M_t\}\leq\varepsilon,
    \qquad
    r_t\leq\eta,
\]
and uncertainty is sufficiently small.  This is a computational version of the same statistical principle: boundary closeness is not enough; boundary closeness must be stable.

\subsection{Adaptive treatment and monitoring}

In dynamic treatment-regime problems, the decision is often not whether one probability is near zero or one, but whether one action has become clearly preferable for the current individual.  Let $a\in\A$ denote a treatment action and let
\[
    M_t^{(a)}=\Pp(Y^{(a)}=1\given H_t)
\]
be the patient-specific potential-outcome risk under action $a$.  A forward sequential rule may recommend $a^*$ when
\[
    \Pp\left\{M_t^{(a^*)}-M_t^{(a)}\geq \Delta
    \text{ for all }a\neq a^*\given\D_t\right\}
    \geq1-\alpha.
\]
The reverse-martingale refinement is to require that this advantage is stable along the latent trajectory, for example $r_t^{(a^*)}\leq\eta$ and, if relevant, $r_t^{(a)}\leq\eta$ for competing actions.  This prevents a treatment switch from being driven only by a transient hidden-state fluctuation.  Related sequential treatment formulations include \citet{murphy2003}, \citet{robins2004}, and \citet{chakraborty2013}.

\section{Discussion}
\label{sec:discussion}

The paper makes one compound point: all-failure Bernoulli runs, logistic separation, and near-degenerate sequential risk trajectories are three faces of the same finite-versus-limiting-law confusion, and the remedy in each case is to require boundary closeness, uncertainty control, and trajectory stability together.  The classical tools (one-sided binomial tests, Clopper--Pearson bounds, confidence sequences, SPRT) correctly handle the parts of the problem they were designed for; the reverse-martingale formulation adds the stability interpretation that those tools do not provide.

The numerical studies confirm that this stability condition is not cosmetic.  In Study~4, a boundary-only rule stops on a transient trajectory with probability 0.999, whereas the three-condition rule reduces this to 0.107 while preserving stopping in the genuinely degenerate scenario.  This separation is the operational payoff of the framework.

The framework is intentionally model-agnostic.  In low-dimensional logistic regression, regularized fits and posterior uncertainty supply the required ingredients directly.  In high-dimensional or sequential prediction systems, the backward-coherence penalty $L_{\mathrm{RM}}$ provides a tractable stability diagnostic.  In all cases the message is the same: $M_\infty\in\{0,1\}$ is a statement about a limiting conditional law, not about any finite estimate, and acting on boundary closeness without stability evidence is premature.

\appendix

\section{Numerical-study settings and reproducibility details}
\label{app:numerical_details}

This appendix gives the complete settings used for the numerical studies in Section~\ref{sec:numerical}.  The purpose of these experiments is diagnostic rather than to tune the best possible classifier.  The reported quantities are intended to show how apparent boundary behavior arises under rare events, sparse covariate information, model flexibility, and sequential risk trajectories.  The code supplied with the manuscript generates the simulation results, figures, \texttt{.csv} summaries, and the \LaTeX{} table snippets used in the paper.

\subsection{Fixed-sample one-sided testing benchmark}
\label{app:onesided_benchmark}

The all-failure stopping threshold in Study~1 is the $s=0$ special case of the exact one-sided binomial test of $H_0:p\geq\varepsilon$, with exact $p$-value $(1-\varepsilon)^n$ and dual Clopper--Pearson upper bound $U_n(0;\alpha)=1-\alpha^{1/n}$; see Remark~\ref{rem:classical}.  This benchmark is not carried over to the logistic and reverse-martingale studies, where boundary behavior arises from covariate separation and trajectory instability rather than a single scalar binomial count.

\subsection{SPRT benchmark for Bernoulli rare-event monitoring}
\label{app:sprt_benchmark}

For completeness, the SPRT for $H_0:p=p_0$ versus $H_1:p=p_1$ ($p_1<p_0$) uses the log-likelihood ratio
\[
    L_n=S_n\log\frac{p_1}{p_0}+(n-S_n)\log\frac{1-p_1}{1-p_0},
\]
with Wald boundaries $a=\log\{\beta/(1-\alpha)\}$ and $b=\log\{(1-\beta)/\alpha\}$ \citep{wald1945,waldwolfowitz1948,siegmund1985}.  On an all-failure path ($S_n=0$), the SPRT favours $H_1$ after
\[
    \tau_{\mathrm{SPRT},0}=\left\lceil\frac{\log\{(1-\beta)/\alpha\}}{\log\{(1-p_1)/(1-p_0)\}}\right\rceil
\]
observations.  This should be used in place of Proposition~\ref{prop:allfailure} whenever the scientific question is genuinely a choice between two pre-specified point hypotheses.  The present paper does not compete with SPRT in that domain.  The boundary-degeneracy framework addresses the complementary setting: composite one-sided statements, logistic separation, and trajectory stability, none of which fits the simple-versus-simple SPRT template.

\subsection{Common numerical conventions}
\label{app:common_numerics}

All simulation studies used fixed pseudo-random seeds so that the numerical results are reproducible.  The logistic-regression studies define an ordinary-logistic fit as unstable if at least one of the following occurs:
\begin{enumerate}[label=(\roman*)]
    \item the sample contains only one response class, so that an ordinary binary logistic fit is not identifiable;
    \item the numerical optimizer fails to converge or issues a convergence warning;
    \item the Euclidean norm of the fitted coefficient vector exceeds 50;
    \item the largest absolute fitted training logit exceeds 30;
    \item more than 1\% of fitted training probabilities are numerically extreme, defined as smaller than $10^{-6}$ or larger than $1-10^{-6}$.
\end{enumerate}
These thresholds are not proposed as universal diagnostic constants.  They are used to make the notion of separation-like numerical behavior explicit and reproducible.  The same thresholds are used across the simulated and real-data logistic experiments, except where stated otherwise.

For synthetic logistic models, the intercept $\beta_0$ was calibrated numerically so that
\[
    \E\{\expit(\beta_0+X^\top\beta)\}=\rho,
\]
where $\rho$ is the target marginal event probability.  Calibration was done by solving the above one-dimensional equation using Monte Carlo integration with $250{,}000$ draws from the corresponding normal linear predictor distribution.  Independent test samples of size $5{,}000$ were used in Studies~2 and~3 for the reported predictive log-loss summaries.

\subsection{Study 1: Bernoulli rare-event experiment}
\label{app:study1_settings}

Study~1 considers iid Bernoulli observations
\[
    Y_i\overset{\mathrm{iid}}{\sim}\operatorname{Bernoulli}(p).
\]
The true success probabilities are
\[
    p\in\{0.10,0.01,0.005\},
\]
with maximum sample sizes
\[
    n_{\max}\in\{100,300,600,1000,3000\}.
\]
The practical-boundary tolerances are
\[
    \varepsilon\in\{0.01,0.005\},
    \qquad \alpha=0.05.
\]
The ordinary estimator and Jeffreys-smoothed estimator are
\[
    \hat p_n=\frac{S_n}{n},
    \qquad
    \tilde p_n^J=\frac{S_n+1/2}{n+1}.
\]
The all-failure practical-zero rule stops after
\[
    n_\varepsilon=\left\lceil\frac{\log(\alpha)}{\log(1-\varepsilon)}\right\rceil
\]
consecutive failures, provided $n_{\max}\geq n_\varepsilon$.  Exact probabilities are used for
\[
    \Pp(\hat p_n=0)=(1-p)^n
    \quad\text{and}\quad
    \Pp(\tau_0\leq n_{\max})=(1-p)^{n_\varepsilon}
\]
when the stopping threshold is reachable.  A Monte Carlo check with $R=50{,}000$ repetitions per setting is included in the script, but the main table reports the exact probabilities.

\subsection{Study 2: Low-dimensional rare-event logistic regression}
\label{app:study2_settings}

Study~2 uses a three-covariate logistic data-generating model:
\[
    X_i\sim N(0,I_3),
    \qquad
    Y_i\given X_i\sim \operatorname{Bernoulli}\{\expit(\beta_0+X_i^\top\beta)\},
\]
where
\[
    \beta=(1.0,-0.7,0.5)^\top.
\]
The intercept $\beta_0$ is calibrated separately for each target prevalence
\[
    \rho\in\{0.10,0.01,0.005\}.
\]
The training sample sizes are
\[
    n\in\{50,100,200,500,1000,2000\},
\]
with $R=1{,}000$ Monte Carlo repetitions per $(\rho,n)$ setting.  The main table reports $n\in\{100,500,2000\}$ for compactness.  For each generated training set, the study records the number of events, the one-class sample indicator, ordinary-logistic instability, coefficient norm, largest absolute training logit, fraction of numerically extreme fitted probabilities, and test-set log loss and Brier score.

Two estimators are compared.  The first is the ordinary logistic maximum-likelihood fit, computed by Newton/IRLS with no penalty.  The second is ridge logistic regression, computed by the same IRLS routine with an $L_2$ penalty of size $\lambda=1$ on the slope coefficients and no penalty on the intercept.  Ridge logistic regression is used as a simple finite regularized comparator; it represents the broader idea that Firth, ridge, or Bayesian regularization prevents divergent finite-sample logistic estimates under separation-like data.

\subsection{Study 3: High-dimensional rare-event logistic regression}
\label{app:study3_settings}

Study~3 increases the number of covariates and uses a sparse signal.  The covariates are generated as
\[
    X_i\sim N(0,I_d),
    \qquad d\in\{20,50\}.
\]
The coefficient vector satisfies
\[
    \beta_j=\frac{2}{\sqrt{5}},\quad j=1,\ldots,5,
    \qquad
    \beta_j=0,
    \quad j>5.
\]
The intercept is calibrated to target prevalence
\[
    \rho\in\{0.01,0.005\}.
\]
The training sample sizes are
\[
    n\in\{500,1000\},
\]
and each $(d,\rho,n)$ setting uses $R=1{,}000$ Monte Carlo repetitions.  An independent test set of size $5{,}000$ is generated for each design point.

The study compares an unpenalized logistic fit and a ridge logistic fit.  The unpenalized fit uses \texttt{scikit-learn} logistic regression with no penalty, the \texttt{lbfgs} solver, maximum 200 iterations, and fitted intercept.  The ridge fit uses an $L_2$ penalty with $C=1$.  The recorded summaries include mean number of events, mean events per variable,
\[
    \operatorname{EPV}=\frac{\#\{i:Y_i=1\}}{d},
\]
one-class rate, ordinary-logistic instability rate, median coefficient norm, median largest absolute training logit, mean fraction of numerically extreme fitted probabilities, median test log loss, median test Brier score, and median absolute calibration error on the test sample.

\subsection{Study 4: Boundary closeness versus reverse-martingale-style stability}
\label{app:study4_settings}

Study~4 is a controlled trajectory simulation designed to isolate the role of the stability condition in the reverse-martingale-style stopping rule.  It does not simulate a full likelihood model.  Instead, it directly generates conditional-risk paths $M_t$ over
\[
    t=1,\ldots,T,
    \qquad T=120,
\]
with $R=5{,}000$ trajectories per scenario.  Three scenarios are used.

For the stable boundary-convergent scenario,
\[
    M_t=\operatorname{clip}\{\expit(-1.5-0.06t)+\xi_t,10^{-5},1-10^{-5}\},
    \qquad
    \xi_t\sim N\left(0,\frac{0.025^2}{t}\right).
\]
For the transient-boundary scenario,
\[
    M_t=\operatorname{clip}\left\{0.12-0.105\exp\left[-\frac{(t-50)^2}{80}\right]+\xi_t,10^{-5},1-10^{-5}\right\},
    \qquad
    \xi_t\sim N(0,0.006^2).
\]
For the stable-interior scenario,
\[
    M_t=\operatorname{clip}\{0.10+0.005\sin(t/10)+\xi_t,10^{-5},1-10^{-5}\},
    \qquad
    \xi_t\sim N(0,0.004^2).
\]
The boundary distance is
\[
    B_t=\min\{M_t,1-M_t\}.
\]
The stability defect uses a five-step future average,
\[
    r_t=\left|M_t-\frac{1}{5}\sum_{k=1}^5 M_{t+k}\right|,
\]
where defined.  The uncertainty-width schedule is
\[
    W_t=\frac{0.10}{\sqrt{t}}.
\]
The boundary-only rule is
\[
    \tau_B=\inf\{t\geq n_{\min}:B_t\leq\varepsilon\},
\]
whereas the boundary-plus-stability rule is
\[
    \tau_{RM}=\inf\{t\geq n_{\min}:B_t\leq\varepsilon,
    \ r_t\leq\eta,
    \ W_t\leq w\}.
\]
The numerical constants are
\[
    \varepsilon=0.02,
    \qquad
    \eta=0.0008,
    \qquad
    w=0.015,
    \qquad
    n_{\min}=10.
\]
The reported summaries are stopping probability, mean stopping time conditional on stopping, and median stopping time conditional on stopping.

\subsection{Study 5: Real-data rare-event illustration using RAND HIE}
\label{app:study5_settings}

Study~5 uses the RAND Health Insurance Experiment data set distributed with \texttt{statsmodels}.  The binary outcome is \texttt{hlthp}, indicating self-rated poor health.  In the version used here, the full data contain $20{,}190$ observations and the event prevalence is approximately $1.50\%$.  The analysis is intended as a real-data rare-event diagnostic rather than a causal analysis of the RAND experiment.

The base design uses six predictors:
\[
    \texttt{lncoins},\ \texttt{idp},\ \texttt{lpi},\ \texttt{fmde},\ \texttt{physlm},\ \texttt{disea}.
\]
These predictors are standardized before fitting.  The quadratic design contains all first- and second-order polynomial terms of the six standardized predictors, excluding the intercept; this gives $d=27$ covariates.  For each design, random training subsamples are drawn without replacement at
\[
    n\in\{100,200,500,1000,2000\}.
\]
There are $R=1{,}000$ random subsamples for each combination of design and sample size.

The study compares two logistic fits.  The column labelled ``MLE'' in the real-data table is a nearly unpenalized logistic regression with a very weak $L_2$ penalty, $C=10^4$, used only to obtain finite numerical optimization in separation-like samples.  The ridge comparator uses $C=1$.  Both are fit with the \texttt{liblinear} solver, fitted intercept, tolerance $10^{-4}$, and maximum 50 iterations.  For each training subsample the recorded quantities are the number of events, event rate, one-class indicator, convergence flag, coefficient norm, largest absolute training logit, fraction of numerically extreme fitted probabilities, mean fitted probability, training log loss, and training Brier score.  The instability criterion is the common criterion in Appendix~\ref{app:common_numerics}.

Because the real-data study repeatedly resamples the same finite data set, the reported log-loss and Brier-score values are training-sample diagnostics, not external predictive-performance estimates.  Their role is to document how rare-event scarcity and model flexibility produce separation-like numerical behavior in a public data example.


\begin{thebibliography}{99}

\bibitem[Albert and Anderson(1984)]{albert1984}
Albert, A. and Anderson, J.\,A. (1984).
On the existence of maximum likelihood estimates in logistic regression models.
\textit{Biometrika}, 71(1), 1--10.
\href{https://doi.org/10.1093/biomet/71.1.1}{doi:10.1093/biomet/71.1.1}

\bibitem[Chakraborty and Moodie(2013)]{chakraborty2013}
Chakraborty, B. and Moodie, E.\,E.\,M. (2013).
\textit{Statistical Methods for Dynamic Treatment Regimes}.
Springer, New York.
\href{https://doi.org/10.1007/978-1-4614-7428-9}{doi:10.1007/978-1-4614-7428-9}

\bibitem[Clopper and Pearson(1934)]{clopper1934}
Clopper, C.\,J. and Pearson, E.\,S. (1934).
The use of confidence or fiducial limits illustrated in the case of the binomial.
\textit{Biometrika}, 26(4), 404--413.
\href{https://doi.org/10.1093/biomet/26.4.404}{doi:10.1093/biomet/26.4.404}

\bibitem[Doob(1953)]{doob1953}
Doob, J.\,L. (1953).
\textit{Stochastic Processes}.
John Wiley \& Sons, New York.

\bibitem[Durrett(2019)]{durrett2019}
Durrett, R. (2019).
\textit{Probability: Theory and Examples} (5th ed.).
Cambridge University Press, Cambridge.
\href{https://doi.org/10.1017/9781108591034}{doi:10.1017/9781108591034}

\bibitem[Firth(1993)]{firth1993}
Firth, D. (1993).
Bias reduction of maximum likelihood estimates.
\textit{Biometrika}, 80(1), 27--38.
\href{https://doi.org/10.1093/biomet/80.1.27}{doi:10.1093/biomet/80.1.27}

\bibitem[Foo and Chang(2026)]{foo_chang_2026_optimal_stopping}
Foo, H.-M. and Chang, Y.-c.\,I. (2026).
Optimal stopping in sequential clinical prediction.
Preprint, \href{https://arxiv.org/abs/2604.22216}{arXiv:2604.22216 [stat.ME]}.

\bibitem[Gelman et~al.(2008)]{gelman2008}
Gelman, A., Jakulin, A., Pittau, M.\,G., and Su, Y.-S. (2008).
A weakly informative default prior distribution for logistic and other regression models.
\textit{The Annals of Applied Statistics}, 2(4), 1360--1383.
\href{https://doi.org/10.1214/08-AOAS191}{doi:10.1214/08-AOAS191}

\bibitem[Heinze and Schemper(2002)]{heinze2002}
Heinze, G. and Schemper, M. (2002).
A solution to the problem of separation in logistic regression.
\textit{Statistics in Medicine}, 21(16), 2409--2419.
\href{https://doi.org/10.1002/sim.1047}{doi:10.1002/sim.1047}

\bibitem[Howard et~al.(2021)]{howard2021}
Howard, S.\,R., Ramdas, A., McAuliffe, J., and Sekhon, J. (2021).
Time-uniform, nonparametric, nonasymptotic confidence sequences.
\textit{The Annals of Statistics}, 49(2), 1055--1080.
\href{https://doi.org/10.1214/20-AOS1991}{doi:10.1214/20-AOS1991}

\bibitem[Lehmann and Romano(2005)]{lehmann2005}
Lehmann, E.\,L. and Romano, J.\,P. (2005).
\textit{Testing Statistical Hypotheses} (3rd ed.).
Springer, New York.
\href{https://doi.org/10.1007/0-387-27605-X}{doi:10.1007/0-387-27605-X}

\bibitem[Murphy(2003)]{murphy2003}
Murphy, S.\,A. (2003).
Optimal dynamic treatment regimes.
\textit{Journal of the Royal Statistical Society, Series B}, 65(2), 331--355.
\href{https://doi.org/10.1111/1467-9868.00389}{doi:10.1111/1467-9868.00389}

\bibitem[Robins(2004)]{robins2004}
Robins, J.\,M. (2004).
Optimal structural nested models for optimal sequential decisions.
In D.\,Y.\ Lin and P.\,J.\ Heagerty (eds.),
\textit{Proceedings of the Second Seattle Symposium in Biostatistics},
Lecture Notes in Statistics, vol.\ 179, pp.\ 189--326. Springer, New York.
\href{https://doi.org/10.1007/978-1-4419-9076-1_11}{doi:10.1007/978-1-4419-9076-1\_11}

\bibitem[Siegmund(1985)]{siegmund1985}
Siegmund, D. (1985).
\textit{Sequential Analysis: Tests and Confidence Intervals}.
Springer, New York.
\href{https://doi.org/10.1007/978-1-4613-9549-7}{doi:10.1007/978-1-4613-9549-7}

\bibitem[Statsmodels Developers(2024)]{statsmodels_randhie}
Statsmodels Developers (2024).
RAND Health Insurance Experiment Data.
\textit{Statsmodels Datasets Documentation}, version~0.14.4.
\url{https://www.statsmodels.org/v0.14.4/datasets/generated/randhie.html}

\bibitem[Wald(1945)]{wald1945}
Wald, A. (1945).
Sequential tests of statistical hypotheses.
\textit{The Annals of Mathematical Statistics}, 16(2), 117--186.
\href{https://doi.org/10.1214/aoms/1177731118}{doi:10.1214/aoms/1177731118}

\bibitem[Wald(1947)]{wald1947}
Wald, A. (1947).
\textit{Sequential Analysis}.
John Wiley \& Sons, New York.

\bibitem[Wald and Wolfowitz(1948)]{waldwolfowitz1948}
Wald, A. and Wolfowitz, J. (1948).
Optimum character of the sequential probability ratio test.
\textit{The Annals of Mathematical Statistics}, 19(3), 326--339.
\href{https://doi.org/10.1214/aoms/1177730197}{doi:10.1214/aoms/1177730197}

\end{thebibliography}
\end{document}